\documentclass[11pt]{article}
\usepackage{amssymb,amsmath,amsthm,amscd,latexsym}
\usepackage{mathrsfs}
\usepackage{mathrsfs}
\usepackage{amsfonts}

\usepackage{amscd}
\usepackage{mathrsfs}
\usepackage{color}
\usepackage{latexsym}
\usepackage{indentfirst}
\usepackage{enumitem}
\usepackage{anysize}
\usepackage{soul}
\usepackage{diagbox}
\usepackage{cancel}

\renewcommand{\paragraph}{\roman{paragraph}}
\input cyracc.def

\newcommand{\Z}{\mathbb{Z}}
\newcommand{\C}{\mathcal{C}}
\newcommand{\K}{\mathcal{K}}
\newcommand{\G}{\mathcal{G}}

\newcommand{\vv}{\mathbf{v}}
\newcommand{\x}{\mathbf{x}}
\newcommand{\y}{\mathbf{y}}

\newtheorem{theorem}{Theorem}

\newtheorem{exmpl}[theorem]{Example}

\newtheorem{lemma}[theorem]{Lemma}
\newtheorem{proposition}[theorem]{Proposition}

\theoremstyle{definition}

\newtheorem{remark}[theorem]{Remark}

\newcommand{\be}{\begin{eqnarray}}
\newcommand{\ee}{\end{eqnarray}}

\begin{document}

\title{$\mathbb{Z}_p\mathbb{Z}_{p^2}$-linear codes: rank and kernel%
\thanks{This research is supported by the National Natural Science Foundation of China (12071001) and the Excellent Youth Foundation of Natural Science Foundation of Anhui Province (1808085J20).}}
\author{ Minjia Shi\thanks{smjwcl.good@163.com}, Shukai Wang\thanks{wangshukai\_2017@163.com}, Xiaoxiao Li\thanks{ahulixiaoxiao@163.com},
\thanks{Key Laboratory of Intelligent Computing
and Signal Processing of Ministry of
Education, School of Mathematics Sciences, Anhui
University, Hefei, Anhui, China.}}
\date{}
\maketitle
\begin{abstract}
{A code $C$ is called $\Z_p\Z_{p^2}$-linear if it is the Gray image of a $\Z_p\Z_{p^2}$-additive code, where $p>2$ is prime.
In this paper, the rank and the dimension of the kernel of $\Z_p\Z_{p^2}$-linear codes are studied.
Two bounds of the rank of a $\Z_3\Z_{9}$-linear code and the dimension of the kernel of a $\Z_p\Z_{p^2}$-linear code are given, respectively.
For each value of these bounds, we give detailed construction of the corresponding code.
Finally, pairs of rank and the dimension of the kernel of $\Z_3\Z_{9}$-linear codes are also considered.
}
\end{abstract}

\section{Introduction}

Let $\mathbb{Z}_p$ and $\mathbb{Z}_{p^2}$ be the ring of integers modulo $p$ and $p^2$, respectively, where $p>2$ is prime.
For a positive integer $m$, $\Z_p^m$ and $\Z_{p^2}^m$ are the ring extension.
The linear codes over $\Z_p$ and $\Z_{p^2}$ are subgroups of $\Z_p^m$ and $\Z_{p^2}^m$, respectively.

The classical Gray map from $\Z_{p^{k+1}}$ to $\Z_p^{p^k}$, where $p$ is a prime and $k\geq 1$, is given in \cite{LS}.
Here, let $k=1$, the Gray map $\phi$ from $\Z_{p^2}$ to $\Z_p^{p}$ is as follows.

$$\begin{array}{ccc}
\Z_{p^2} & \longrightarrow & \Z_p^p\\
\theta & \longmapsto & \phi(\theta)
\end{array},$$
$\phi(\theta)=\theta''(1,1,\ldots,1)+ \theta' (0,1,\ldots,p-1)$, where
$\theta = \theta'' p + \theta',$ and $\theta',\theta''\in\{0,\ldots,p-1\}.$

The homogeneous weight  \cite{LS} $wt_{hom}$ on $\Z_{p^2}$ is $$
wt_{hom}(x)=\left\{
\begin{array}{ll}
p,& \text{ if } x\in p\Z_{p^2}\backslash\{0\}, \\
p-1,& \text{ if } x\notin p\Z_{p^2}, \\
0,& \text{ if } x=0.
\end{array}
\right.
$$
And the homogeneous distance is $d_{hom}=wt_{hom}(\mathbf{a}-\mathbf{b})$, where $\mathbf{a},\mathbf{b}\in \Z_{p^2}^n$.
Then, $\phi$ is an isometry from $(\Z_{p^2}^n, d_{hom})$ to $(\Z_p^{pn}, d_H)$, where $d_H$ is the Hamming distance on $\Z_p^{p}$.

If $\mathcal{C}$ is a linear code over $\mathbb{Z}_{p^2}$, then the code $C=\phi(\mathcal{C})$ is called a \emph{$\mathbb{Z}_{p^2}$-linear} code.
The dual of a linear code $\C$ of length $n$ over $\mathbb{Z}_{p^2}$, denoted by $\C^{\bot}$, is defined as
$$\mathcal{C}^\bot=\{\mathbf{x}\in \Z_{p^2}^n:\langle\mathbf{x},\mathbf{y}\rangle=0 ~\rm {for ~all}~ \textbf{y}\in \mathcal{C}\},$$
where $\langle,\rangle$ denotes the usual Euclidean inner product.
The code $C_\bot=\phi(\C^{\bot})$ is called the \emph{$\mathbb{Z}_{p^2}$-dual code} of $C=\phi(\C)$.

In 1973, Delsarte first defined the additive codes in terms of association schemes \cite{DP}, it is a subgroup of the underlying abelian group.
Borges et al. \cite{BFPR10} studied the standard generator matrix and the duality of $\Z_2\Z_4$-additive codes. Since then, a lot of work has been devoted to characterizing $\Z_2\Z_4$-additive codes. Dougherty et al. \cite{DLY16} constructed one weight $\Z_2\Z_4$-additive codes and analyzed their parameters. Benbelkacem et al. \cite{BBDF20} studied $\Z_2\Z_4$-additive complementary dual codes and their Gray images. In fact, these codes can be viewed as a generalization of the linear complementary dual (LCD for short) codes \cite{M92} over finite fields. Joaquin et al. \cite{BBCM} introduced a decoding method of the $\Z_2\Z_4$-linear codes.
More structure properties of $\Z_2\Z_4$-additive codes can be found in \cite{BBDF11,JTCR2}. Moreover, the additive codes over different mixed alphabet have also been intensely studied, for example $\Z_2\Z_2[u]$-additive codes \cite{BC2}, $\Z_2\Z_{2^s}$-additive codes \cite{AS13}, $\Z_{p^r}\Z_{p^s}$-additive codes \cite{AS15}, $\Z_2\Z_2[u,v]$-additive codes \cite{SW}, $\Z_p\Z_{p^k}$-additive codes \cite{SWD} and $\Z_p(\Z_p+u\Z_p)$-additive codes \cite{WS}, and so on. It is worth mentioning that $\Z_2\Z_4$-additive cyclic codes form an important family of $\Z_2\Z_4$-additive codes, many optimal binary codes can be obtained from the images of this family of codes. More details of $\Z_2\Z_4$-additive cyclic codes can be found in \cite{BC,JCR,JTCR,JTCR2,YZ20}.

In \cite{CFC}, Fern\'{a}ndez et al. studied the rank and kernel of $\Z_2\Z_4$-linear codes, where $\Z_2\Z_4$-linear codes are $\Z_2\Z_4$-additive codes under the generalized Grap map $\Phi$.
The authors also studied the rank and kernel of $\Z_2\Z_4$-additive cyclic codes in \cite{JTCR2}.
It is a natural problem: can we study the rank and kernel of $\Z_p\Z_{p^2}$-linear codes in terms of $\Z_p\Z_{p^2}$-additive codes?
In this paper, we succeed in considering this general case for the kernel, and we also generalize the rank over $\Z_3\Z_9$.

Let $\C$ be a $\Z_p\Z_{p^2}$-additive code of length $\alpha+\beta$ defined by a subgroup of $\Z_p^\alpha\times\Z_{p^2}^\beta$. Let $n=\alpha+p\beta$ and
$\Phi:\Z_p^\alpha\times\Z_{p^2}^\beta\rightarrow \Z_p^n$ is an extension of the Gray map, which
$$\Phi(\x,\y)=(\x,\phi(y_1),\ldots,\phi(y_\beta)),$$
for any $\x\in \Z_p^\alpha$, and $\y=(y_1,\ldots,y_\beta)\in \Z_{p^2}^\beta$. The code $C=\Phi(\C)$ is called a $\Z_p\Z_{p^2}$\emph{-linear code}.
There are two parameters of this nonlinear code: the rank and the dimension of the kernel.
We denote $\langle C\rangle$ the linear span of the codewords of $C$. The dimension of $\langle C\rangle$ is called the \emph{rank} of the code $C$, denoted by $rank(C)$.
The \emph{kernel} of the code $C$, which is denoted as $K(C)$, is defined as:
$$K(C)=\{\x\in\Z_p^{n}|C+\x=C\},$$
where $C+\x$ means $\x$ adds all codewords in $C$.
We will denote the dimension of the kernel of $C$ by \emph{$ker(C)$}.

The rank and dimension of the kernel have been studied for some families of $\Z_2\Z_4$-linear codes \cite{BPR,FPV,PRV}.
These two parameters are helpful to the classification of $\Z_2\Z_4$-linear codes.
Therefore, we try to generalize them to the more general case over $\Z_p\Z_{p^2}$.

\par
The paper is organized as follows.
In Section \ref{sec:2}, we give some properties about $\mathbb{Z}_p\mathbb{Z}_{p^2}$-additive codes and $\mathbb{Z}_p\mathbb{Z}_{p^2}$-linear codes.
In Section \ref{sec:3}, we find all values of the rank for $\mathbb{Z}_3\mathbb{Z}_{9}$-linear codes, and we construct a $\mathbb{Z}_3\mathbb{Z}_{9}$-linear code for each value.
In Section \ref{sec:4}, we determine all values of the dimension of the kernel for $\mathbb{Z}_p\mathbb{Z}_{p^2}$-linear codes.
We also construct all $\mathbb{Z}_p\mathbb{Z}_{p^2}$-linear codes for the values of the dimension of the kernel.
In Section \ref{sec:5}, pairs of rank and the dimension of kernel of $\mathbb{Z}_3\mathbb{Z}_{9}$-additive codes are studied.
For each fixed value of the dimension of the kernel, the range of rank is given.
Moreover, the construction method of the $\mathbb{Z}_3\mathbb{Z}_{9}$-linear codes with pairs of rank and the dimension of the kernel is provided.
In Section \ref{sec:6}, we conclude the paper.

\section{Preliminaries} \label{sec:2}

Let $\mathcal{C}$ be a $\mathbb{Z}_p\mathbb{Z}_{p^2}$-additive code. Then $\mathcal{C}$ is isomorphic to an abelian structure $\mathbb{Z}_p^\gamma \times \mathbb{Z}_{p^2}^\delta$ since it is a subgroup of $\mathbb{Z}_p^\alpha \times \mathbb{Z}_{p^2}^\beta$.
The \emph{order} of a codeword $\mathbf{c}$ means the minimal positive integer $a$ such that $a\cdot\mathbf{c}=\mathbf{0}$.
Therefore, the size of $\mathcal{C}$ is $p^{\gamma+2\delta}$, and the number of codewords of order $p$ in $\mathcal{C}$ is $p^{\gamma+\delta}$.

Let $X$ be the set of $\mathbb{Z}_p$ coordinate positions, and $Y$ be the set of $\mathbb{Z}_{p^2}$ coordinate positions, so $|X|=\alpha$ and $|Y|=\beta$.
In general, it is notice that the first $\alpha$ positions corresponds to the set $X$ and the last $\beta$ positions corresponds to the set $Y$.
We denote $\mathcal{C}_X$ and $\mathcal{C}_Y$ are the punctured codes of $\mathcal{C}$ by deleting the coordinates outside $X$ and $Y$, respectively.
Let $\mathcal{C}_p$ be the subcode consisting of all codewords of order $p$ in $\mathcal{C}$. Let $\kappa$ be the dimension of the linear code $(\mathcal{C}_p)_X$ over $\Z_p$.
If $\alpha=0$, then $\kappa=0$.
Considering all these parameters, we will say that $\mathcal{C}$ or $C=\Phi(\mathcal{C})$ is of type $(\alpha, \beta; \gamma, \delta; \kappa)$.
\par
For a $\mathbb{Z}_p\mathbb{Z}_{p^2}$-additive code, every codeword $\mathbf{c}$ can be uniquely expressible in the form
$$\mathbf{c}=\sum_{i=1}^\gamma \lambda_i\mathbf{u}_i+\sum_{j=1}^\delta \nu_j\mathbf{v}_j,$$
where $\lambda_i \in \mathbb{Z}_p$ for $1\leq i \leq \gamma$, $\nu_j \in \mathbb{Z}_{p^2}$ for $1 \leq j \leq \delta$ and $\mathbf{u}_i$, $\mathbf{v}_j$ are vectors in $\mathbb{Z}_p^\alpha \times \mathbb{Z}_{p^2}^\beta$ of order $p$ and $p^2$, respectively.
Then, we get the generator matrix $\mathcal{G}$ for the code $\mathcal{C}$ by vectors $\mathbf{u}_i$ and $\mathbf{v}_j$.
$$\mathcal{G}=\left( \begin{array}{c}\mathbf{u}_1\\ \vdots \\\mathbf{u}_\gamma \\ \mathbf{v}_1 \\\vdots \\\mathbf{v}_\delta \end{array}\right)=\left( \begin{array}{c|c}B_1 & pB_3 \\
 \hline B_2 & Q\end{array}\right),$$
where $B_1$, $B_2$ and $B_3$ are matrices over $\mathbb{Z}_p$ of size $\gamma \times \alpha, \delta \times \alpha$ and $\gamma \times \beta$, respectively; and $Q$ is a matrix over $\mathbb{Z}_{p^2}$ of size $\delta\times \beta$.

Recall that in coding theory, two linear codes $C_1$ and $C_2$ of length $n$ are \emph{permutation equivalent}
if there exists a coordinate permutation $\pi$ such that $C_2=\{\pi(c)|c\in C_1\}$.
The permutation equivalent of $\Z_p\Z_{p^2}$-additive codes can be defined similarly.
Then we have the following theorem.

\begin{theorem}{\rm\cite{AS15}} \label{th:1}
Let $\mathcal{C}$ be a $\mathbb{Z}_p\mathbb{Z}_{p^2}$-additive code of type $(\alpha,\beta;\gamma,\delta;\kappa)$. Then, $\mathcal{C}$ is permutation equivalent to a $\mathbb{Z}_p\mathbb{Z}_{p^2}$-additive code $\mathcal{C}^\prime$ with generator matrix of the form:
$$\mathcal{G}_S=\left(\begin{array}{cc|ccc}I_\kappa & T' & pT_2 & \mathbf{0} & \mathbf{0}\\
\mathbf{0} & \mathbf{0} & pT_1 & pI_{\gamma-\kappa} & \mathbf{0}\\
\hline
\mathbf{0} & S' & S & R & I_\delta
\end{array}\right),$$
where $I_\delta$ is the identity matrix of size $\delta\times \delta$; $T', S', T_1, T_2$ and $R$ are matrices over $\mathbb{Z}_p$; and $S$ is a matrix over $\mathbb{Z}_{p^2}$.
\end{theorem}
\begin{proof}
Straightforward by setting $r=1,s=2$ of $\Z_{p^r}\Z_{p^s}$-linear codes in \cite{AS15}.
\end{proof}

From Theorem \ref{th:1}, there is a $\mathbb{Z}_p \mathbb{Z}_{p^2}$-additive code $\mathcal{C}$ of type $(\alpha,\beta;\gamma,\delta;\kappa)$ if and only if
\begin{equation}\label{eq:1}
\begin{array}{c}
\alpha,\beta,\gamma,\delta,\kappa\geq 0, \alpha+\beta>0,\\
0<\delta+\gamma \leq \beta+\kappa \text{\rm{ and }} \kappa\leq \rm{min}(\alpha,\gamma).
\end{array}
\end{equation}

The definition of the duality for $\mathbb{Z}_p\mathbb{Z}_{p^k}$-additive codes is shown in \cite{SWD}, which includes $\mathbb{Z}_p\mathbb{Z}_{p^2}$-additive codes.
The \emph{inner product} between $(\mathbf{v}_1|\mathbf{w}_1)$ and $(\mathbf{v}_2|\mathbf{w}_2)$ in $\mathbb{Z}_p^\alpha \times \Z_{p^2}^\beta$ can be written as follows:
$$\langle(\mathbf{v}_1|\mathbf{w}_1),(\mathbf{v}_2|\mathbf{w}_2)\rangle=p\langle \mathbf{v}_1,\mathbf{v}_2\rangle + \langle \mathbf{w}_1,\mathbf{w}_2\rangle \in \Z_{p^2}.$$
Note that the result of the inner product $\langle \mathbf{v}_1,\mathbf{v}_2\rangle$ is from $\Z_p$,
and multiplication of its value by $p$ should be formally understood as the natural homomorphism from $\Z_p$ into $\Z_{p^2}$,
that is, $p\langle \mathbf{v}_1,\mathbf{v}_2\rangle \in p\Z_{p^2} \subseteq \Z_{p^2}$. More detail are suggested to see \cite{SWD}.

The \emph{dual code} $\mathcal{C}^\bot$ of a $\mathbb{Z}_p\Z_{p^2}$-additive code $\mathcal{C}$ is defined in the standard way by
\begin{eqnarray*}
 \mathcal{C}^\bot=\big\{(\mathbf{x}|\mathbf{y})
 {}\in \mathbb{Z}_p^\alpha \times \Z_{p^2}^\beta:
 {
 \langle(\mathbf{x}|\mathbf{y}),(\mathbf{v}|\mathbf{w})\rangle=0 ~\rm {for ~all}~ (\mathbf{v}|\mathbf{w})\in \mathcal{C}\big\}.}
\end{eqnarray*}
Readily, the dual code $\mathcal{C}^\bot$ is also a $\mathbb{Z}_p\Z_{p^2}$-additive code \cite{AS15}.
In particular, it is shown that the type of the code $\C^\bot$ is $(\alpha, \beta; \bar{\gamma}, \bar{\delta}; \bar{\kappa})$, where
$$\begin{array}{l}\bar{\gamma}=\alpha+\gamma-2\kappa, \\
\bar{\delta}=\beta-\gamma-\delta+\kappa, \\
\bar{\kappa}=\alpha-\kappa.
\end{array}
$$
The code $\Phi(\mathcal{C}^\bot)$ is denoted by $C_\bot$ and called the \emph{$\mathbb{Z}_p \mathbb{Z}_{p^2}$-dual code} of $C$.

\begin{lemma}\label{l:2}
For all $\mathbf{u},\mathbf{v}\in \mathbb{Z}_p^\alpha \times \mathbb{Z}_{p^2}^\beta$,
$\mathbf{u}=(u_1,\ldots,u_{\alpha+\beta})$,
$\mathbf{v}=(v_1,\ldots,v_{\alpha+\beta})$,
we have
$$
\Phi(\mathbf{u}+\mathbf{v})=\Phi(\mathbf{u})+\Phi(\mathbf{v})+\Phi(pP(\mathbf{u},\mathbf{v})),
$$
where $P(\mathbf{u},\mathbf{v})=(0,\ldots,0,P(u_{\alpha+1},v_{\alpha+1}),\ldots,P(u_{\alpha+\beta},v_{\alpha+\beta}))$ and
$$
P(u_i,v_i)=
P(u'_i,v'_i)=
\begin{cases}
 1 &
 \mbox{if $u'_i+v'_i \ge p$}
 \\ 0 &\mbox{otherwise},
             \end{cases}
\qquad u_i=u''_i p + u'_i, \quad  v_i=v''_i p + v'_i.
$$
\end{lemma}
\begin{proof}
It is sufficient to prove the claim for one
$\mathbb{Z}_{p^2}$ coordinate.
Assume
 $w_i 
= u_i + v_i$ (the addition is modulo $p^2$),
where
$w_i=w''_i p + w'_i$,
$u_i=u''_i p + u'_i$, and
$v_i=v''_i p + v'_i$.
Since
$w_i 
= u_i + v_i = (u_i''+v_i'')p + (u_i'+v_i') \bmod p^2 $,
we have $w'_i = u'_i+v'_i$ and $w''_i = u''_i+v''_i \bmod p$
if $u'_i+v'_i < p$. If $u'_i+v'_i \ge p$,
the formula is different:
$w'_i = u'_i+v'_i \bmod p$ and $w''_i = u''_i+v''_i + 1 \bmod p$.
Utilizing the definition of $P$, we get
\begin{equation}\label{eq:'''}
w'_i = u'_i+v'_i \bmod\ p \qquad\mbox{and}\quad
w''_i = u''_i+v''_i+P(u'_i,v'_i) \bmod\ p.
\end{equation}
Now, from the definition of $\phi$ and \eqref{eq:'''}, it is straightforward
$ \phi( pP(u'_i,v'_i)) = P(u'_i,v'_i) (1,1,\ldots,1),$
and we find
\begin{multline}\label{eq:phiu+v}
 \phi(u_i + v_i) =
 ( u''_i+v''_i+P(u'_i,v'_i) ) (1,1,\ldots,1)
+ (u'_i+v'_i) (0,1,\ldots,p-1)\\
= \phi(u_i) + \phi(v_i) + \phi( pP(u'_i,v'_i)).
\end{multline}
 Applying \eqref{eq:phiu+v} to each coordinate from
 $\alpha+1$ to $\alpha+\beta$ completes the proof.
\end{proof}
\begin{remark}
 The function $P$ can be treated as
 a $\{0,1,\ldots,p-1\}$-valued function in
 two $\{0,1,\ldots,p-1\}$-valued arguments
 $u'_i$, $v'_i$.
 As any such function, it can be represented
 as a polynomial of degree at most $p-1$ in each variable. For example $P(u'_i,v'_i) = u'_i v'_i$
 for $p=2$ and
 $P(u'_i,v'_i) = 2u'_i v'_i(1 + u'_i + v'_i) $
 for $p=3$. We also note that substituting
 $u_i$ and $v_i$ instead of $u'_i$ and $v'_i$
 in this polynomial does not change the value
 of $p P(u_i,v_i)$; so, for example,
 for $p=3$, it is safe to write
 $P(u_i,v_i) = 2u_i v_i(1 + u_i + v_i) $.
 For the rigor of the paper, we still use a new function $P'(\mathbf{u},\mathbf{v})$,
 where $P(\mathbf{u},\mathbf{v})=(p-1)P'(\mathbf{u},\mathbf{v})$.
 Note that $\Phi(pP(\mathbf{u},\mathbf{v}))=(p-1)\Phi(pP'(\mathbf{u},\mathbf{v}))$.
\end{remark}

\begin{lemma}\label{l:4}
Let $\mathcal{C}$ be a $\mathbb{Z}_p \mathbb{Z}_{p^2}$-additive code. The $\mathbb{Z}_p \mathbb{Z}_{p^2}$-linear code $C=\Phi(\mathcal{C})$ is linear if and only if
$pP'(\mathbf{u},\mathbf{v})\in \mathcal{C}$ for all $\mathbf{u},\mathbf{v}\in \mathcal{C}$.
\end{lemma}

\begin{proof}
Assume that $C$ is linear. Because $\mathcal{C}$ is additive, for any $\mathbf{u},\mathbf{v}\in \mathcal{C}$, we have $\mathbf{u}+\mathbf{v}\in \mathcal{C}$.
Then, $\Phi(\mathbf{u}),\Phi(\mathbf{v}),\Phi(\mathbf{u+v})\in C$.
Since $C$ is linear, $\Phi(pP'(\mathbf{u},\mathbf{v}))\in C$ by Lemma \ref{l:2}.
Therefore, $pP'(\mathbf{u},\mathbf{v})\in \mathcal{C}$. Conversely, assume that $pP'(\mathbf{u},\mathbf{v})\in \mathcal{C}$ for all $\mathbf{u},\mathbf{v}\in \mathcal{C}$. Let $\mathbf{x},\mathbf{y}\in C$, then there are $\mathbf{u'},\mathbf{v'}\in \C$ such that $\mathbf{x}=\phi(\mathbf{u'}),\mathbf{y}=\Phi(\mathbf{v'})$.
Thus $pP'(\mathbf{u'},\mathbf{v'})\in \mathcal{C}$.
Since $\C$ is additive, $\mathbf{u'}+\mathbf{v'}+pP'(\mathbf{u'},\mathbf{v'})\in \mathcal{C}$.
Therefore, $\Phi(\mathbf{u'}+\mathbf{v'}+pP'(\mathbf{u'},\mathbf{v'}))\in C$.
Then $$\begin{aligned}
\Phi(\mathbf{u'}+\mathbf{v'}+pP'(\mathbf{u'},\mathbf{v'}))
&=\Phi(\mathbf{u'}+\mathbf{v'})+\Phi(pP'(\mathbf{u'},\mathbf{v'}))\\
&=\Phi(\mathbf{u'})+\Phi(\mathbf{v'})+\Phi(pP(\mathbf{u'},\mathbf{v'}))+\Phi(pP'(\mathbf{u'},\mathbf{v'}))\\
&=\Phi(\mathbf{u'})+\Phi(\mathbf{v'})+(p-1)\Phi(pP'(\mathbf{u'},\mathbf{v'}))+\Phi(pP'(\mathbf{u'},\mathbf{v'}))\\
&=\Phi(\mathbf{u'})+\Phi(\mathbf{v'}).
\end{aligned}$$
Hence, $\Phi(\mathbf{u'})+\Phi(\mathbf{v'})=\x+\y \in C$.
\end{proof}

\section{Rank of $\mathbb{Z}_3 \mathbb{Z}_{9}$-additive codes}\label{sec:3}

Let $\mathcal{C}$ be a $\mathbb{Z}_p \mathbb{Z}_{p^2}$-additive code of type $(\alpha, \beta;\gamma,\delta;\kappa)$ and $C=\Phi(\mathcal{C})$ of length $\alpha+p\beta$.
We have known the definition of $rank(C)$ before. For general prime $p$, it is very difficult to discuss all values of $r=rank(C)$ clearly.
Therefore, in this section, we only consider $p=3$, then give the range of values $r=rank(C)$ and prove that there is a $\mathbb{Z}_3 \mathbb{Z}_{9}$-linear code of type $(\alpha, \beta;\gamma,\delta;\kappa)$ with $r=rank(\mathcal{C})$ for any positive integer $r$.
%

If $p=3$, for all $\mathbf{u},\mathbf{v}\in \mathbb{Z}_3^\alpha \times \mathbb{Z}_{9}^\beta$,
it is easy to check that $\Phi(\mathbf{u}+\mathbf{v})=\Phi(\mathbf{u})+\Phi(\mathbf{v})+2\Phi(3(\mathbf{u}\ast \mathbf{v}+\mathbf{u}\ast \mathbf{u}\ast \mathbf{v}+\mathbf{u}\ast \mathbf{v}\ast \mathbf{v}))$ by Lemma \ref{l:2}, where $\ast$ denotes the componentwise multiplication.
Then we have the following theorem.

\begin{theorem}\label{th:5}
Let $\mathcal{C}$ be a $\mathbb{Z}_3\mathbb{Z}_{9}$-additive code of type $(\alpha, \beta;\gamma,\delta;\kappa)$ which satisfies (\ref{eq:1}),
$C=\Phi(\mathcal{C})$ be the corresponding $\Z_3\Z_{9}$-linear code of length $n=\alpha+3\beta$.

\begin{itemize}
\item[(i)] Let $\mathcal{G}$ be the generator matrix of $\mathcal{C}$,
and let $\{\mathbf{u}_i\}_{i=1}^\gamma$, $\{\mathbf{v}_j\}_{j=1}^\delta$ be the rows of order $3$ and $9$ in $\mathcal{G}$, respectively.
Then $\langle C\rangle$ is generated by $\{\Phi(\mathbf{u}_i)\}_{i=1}^\gamma$, $\{\Phi(\mathbf{v}_j)\}_{j=1}^\delta$, $\{\Phi(3\mathbf{v}_k*\mathbf{v}_l)\}_{1\leq l\leq k\leq \delta}$
and $\{\Phi(3\mathbf{v}_x*\mathbf{v}_y*\mathbf{v}_z)\}_{1\leq x\leq y\leq z\leq\delta}$.
\item[$(ii)$] $rank(C)\in \bigg\{\gamma+2\delta,...,\rm{min}\bigg(\beta+\gamma+\kappa,\gamma+\delta+\dbinom{\delta+1}{2}+\dbinom{\delta+2}{3}\bigg)\bigg\}$.
Let $rank(C)=r=\gamma+2\delta+\overline{r}$. Then $\overline{r}\in\bigg\{0,1,...,\rm{min}\bigg(\beta-(\gamma-\kappa)-\delta,\dbinom{\delta+1}{2}+\dbinom{\delta+2}{3}-\delta\bigg)\bigg\}$.
\item[$(iii)$] The linear code $\langle C\rangle$ over $\Z_3$ is $\mathbb{Z}_3 \mathbb{Z}_{9}$-linear.
\end{itemize}
\end{theorem}
\begin{proof}
$(i)$
Let $\mathbf{c}\in \C$, without loss of generality, $\mathbf{c}$ can be expressed as $\mathbf{c}=\Sigma_{j=1}^\zeta \mathbf{v}_j+\omega$,
where $\zeta\leq\delta$ and $\omega$ is a codeword in $\C$ of order $3$.
By Lemma \ref{l:2}, we have $\Phi(\mathbf{c})=\Phi(\Sigma_{j=1}^\zeta \mathbf{v}_j)+\Phi(\omega)$,
where $\Phi(\omega)$ is a linear combination of $\{\Phi(\mathbf{u}_i)\}_{i=1}^\gamma$, $\{\Phi(3\mathbf{v}_j)\}_{j=1}^\delta$.
Note that
$\Phi(\Sigma_{j=1}^\zeta \mathbf{v}_j)=\Sigma_{j=1}^\zeta\Phi(\mathbf{v}_j)+\Sigma_{1\leq l \leq k\leq \zeta}(\tilde{a}\Phi(3(\mathbf{v}_k\ast v_l))+\Sigma_{1\leq x \leq y\leq z\leq \zeta}(\tilde{b}\Phi(3(\mathbf{v}_x\ast \mathbf{v}_y\ast \mathbf{v}_z))$, $\tilde{a},\tilde{b}\in \Z_3$.
Since $\Phi(3\mathbf{v}_j)=\Phi(3\mathbf{v}_j\ast\mathbf{v}_j\ast\mathbf{v}_j)$ by Lemma \ref{l:2},
$\Phi(\mathbf{c})$ is generated by $\{\Phi(\mathbf{u}_i)\}_{i=1}^\gamma$, $\{\Phi(\mathbf{v}_j)\}_{j=1}^\delta$,
$\{\Phi(3\mathbf{v}_k*\mathbf{v}_l)\}_{1\leq l\leq k\leq \delta}$
and $\{\Phi(3\mathbf{v}_x*\mathbf{v}_y*\mathbf{v}_z)\}_{1\leq x\leq y\leq z\leq\delta}$.

$(ii)$ The bound $\gamma+\delta+\dbinom{\delta+1}{2}+\dbinom{\delta+2}{3}$ is straightforward by (i),
and the bound $\beta+\gamma+\kappa$ is trivial by the form of $\mathcal{G}_S$ in Theorem \ref{th:1}.

$(iii)$ Let $\mathcal{S}_\mathcal{C}$ be a $\Z_3\Z_9$-additive code generated by
$\{\mathbf{u}_i\}_{i=1}^\gamma$, $\{\mathbf{v}_j\}_{j=1}^\delta$, $\{3\mathbf{v}_k*\mathbf{v}_l\}_{1\leq l\leq k\leq \delta}$
and $\{3\mathbf{v}_x*\mathbf{v}_y*\mathbf{v}_z\}_{1\leq x\leq y\leq z\leq\delta}$,
then $\mathcal{S}_\mathcal{C}$ is of type $(\alpha,\beta;\gamma+\overline{r},\delta;\kappa)$ and $\Phi(\mathcal{S}_\mathcal{C})=\langle C\rangle$.
\end{proof}
For convenience, we use $\mathbf{v}_j^2$ to denote $\mathbf{v}_j*\mathbf{v}_j$, and write the vectors in Theorem \ref{th:5} (i) in the following way:
$\{\Phi(\mathbf{u}_i)\}_{i=1}^\gamma$, $\{\Phi(\mathbf{v}_j)\}_{j=1}^\delta$, $\{\Phi(3\mathbf{v}_j)\}_{j=1}^\delta$, $\{\Phi(3\mathbf{v}_j^2)\}_{j=1}^\delta$, $\{\Phi(3\mathbf{v}_k*\mathbf{v}_l)\}_{1\leq l< k\leq \delta}$, $\{\Phi(3\mathbf{v}_x^2*\mathbf{v}_y)\}_{1\leq x< y\leq \delta}$, $\{\Phi(3\mathbf{v}_x*\mathbf{v}_y^2)\}_{1\leq x< y\leq \delta}$, $\{\Phi(3\mathbf{v}_x*\mathbf{v}_y*\mathbf{v}_z)\}_{1\leq x< y<z\leq \delta}$.
The number $\gamma+\delta+\dbinom{\delta+1}{2}+\dbinom{\delta+2}{3}$ can be represented as
$\gamma+2\delta+\dbinom{\delta}{2}+\delta+2\dbinom{\delta}{2}+\dbinom{\delta}{3}$=
$\gamma+3\delta+3\dbinom{\delta}{2}+\dbinom{\delta}{3}$.

The following theorem provides a method to construct a $\Z_3\Z_{9}$-linear code for each value of $rank(C)$ which is shown in Theorem \ref{th:5}.
\begin{theorem}\label{th:7}
Let $\alpha,\beta,\gamma,\delta,\kappa$ be positive integers satisfying (\ref{eq:1}).
Then, there is a $\mathbb{Z}_3 \mathbb{Z}_{9}$-linear code $C$ of type $(\alpha,\beta;\gamma,\delta;\kappa)$ with $rank(C)=r$ if and only if
$$r\in \bigg\{\gamma+2\delta,\ldots,\rm{min}\bigg(\beta+\delta+\kappa,\gamma+\delta+\dbinom{\delta+1}{2}+\dbinom{\delta+2}{3}\bigg)\bigg\}.$$
\end{theorem}
\begin{proof}
Let $\C$ be a $\Z_3\Z_{9}$-additive code of type $(\alpha,\beta;\gamma,\delta;\kappa)$ with generator matrix
$$\mathcal{G}=\left(\begin{array}{cc|ccc}I_\kappa & T' & \mathbf{0} & \mathbf{0} & \mathbf{0}\\
\mathbf{0} & \mathbf{0} & 3T_1 & 3I_{\gamma-\kappa} & \mathbf{0}\\
\hline
\mathbf{0} & S' & S_r & \mathbf{0} & I_\delta
\end{array}\right),$$
where $S_r$ is a matrix over $\Z_{9}$ of size $\delta\times(\beta-(\gamma-\kappa)-\delta)$, and $C=\Phi(\C)$ is a $\Z_3\Z_{9}$-linear code.
Let $\{\mathbf{u}_i\}_{i=1}^\gamma$ and $\{\mathbf{v}_j\}_{j=1}^\delta$ be the row vectors of order $3$ and $9$ in $\G$, respectively.

The necessary condition follows from Theorem \ref{th:5}.
For the sufficiency of the theorem, we should construct a $\mathbb{Z}_3 \mathbb{Z}_{9}$-linear code $C$ of type $(\alpha,\beta;\gamma,\delta;\kappa)$ with $rank(C)=r$.

Let $\mathbf{h}_k,1\leq k\leq \delta$ be the column vector of length $\delta$, with a $1$ in the $k$-th coordinate and zeros elsewhere.
Define five matrices $\bar{A}_{\delta \times \delta}$, $\bar{B},\bar{C},\bar{D}$ of size $\delta \times \dbinom{\delta}{2}$ and $\bar{E}$ of size $\delta \times \dbinom{\delta}{3}$ over $\Z_9$ as follows:
\begin{itemize}
\item $\bar{A}_{\delta \times \delta}=2I_\delta$;
\item $\bar{B}=(\mathbf{b}_j)$, where $\mathbf{b}_j=\mathbf{h}_k+\mathbf{h}_l$ for $1\leq k< l\leq \delta, j\in\bigg\{1,2,\ldots, \dbinom{\delta}{2}\bigg\}$,
and $\mathbf{b}_{j_1}\neq \mathbf{b}_{j_2}$ if $j_1 \neq j_2$;
\item $\bar{C}=2\bar{B}$;
\item $\bar{D}=(\mathbf{d}_j)$, where $\mathbf{d}_j=\mathbf{h}_k+2\mathbf{h}_l$ (or equivalently, $\mathbf{d}_j=2\mathbf{h}_k+\mathbf{h}_l$) for $1\leq k< l\leq \delta, j\in\bigg\{1,2,\ldots, \dbinom{\delta}{2}\bigg\}$,
and $\mathbf{d}_{j_1}\neq \mathbf{d}_{j_2}$ if $j_1 \neq j_2$;
\item $\bar{E}=(\mathbf{e}_j)$, where $\mathbf{e}_j=\mathbf{h}_x+\mathbf{h}_y+\mathbf{h}_z$ for $1\leq x< y<z\leq \delta, j\in\bigg\{1,2,\ldots, \dbinom{\delta}{3}\bigg\}$,
and $\mathbf{e}_{j_1}\neq \mathbf{e}_{j_2}$ if $j_1 \neq j_2$.
\end{itemize}

Since the bound $\beta+\delta+\kappa$ is trivial, it is enough to consider
$\rm{min}\bigg(\beta+\delta+\kappa,\gamma+\delta+\dbinom{\delta+1}{2}+\dbinom{\delta+2}{3}\bigg)=\gamma+\delta+\dbinom{\delta+1}{2}+\dbinom{\delta+2}{3}$,
that is,
$\overline{r}\in \bigg\{0,1,\ldots,\delta+3\dbinom{\delta}{2}+\dbinom{\delta}{3}\bigg\}$.
We construct $S_r$ over $\mathbb{Z}_{9}$ in the following way such that we can get a $\Z_3\Z_{9}$-linear code.

$(1)$ If $\overline{r}\in \{0,1,\ldots,\delta\}$,
then choose $\overline{r}$ different column vectors from the matrix $\bar{A}$ as $\overline{r}$ columns in $S_r$, and the remaining columns are all-zero vectors.
That is, $S_r=(\bar{\bar{A}}|\mathbf{0})$, where $\bar{\bar{A}}$ is a matrix consisting of $\overline{r}$ vectors from $\bar{A}$.
In this case, $\{\Phi(3\mathbf{v}_k*\mathbf{v}_l)\}_{1\leq l< k\leq \delta}$, $\{\Phi(3\mathbf{v}_x^2*\mathbf{v}_y)\}_{1\leq x< y\leq \delta}$, $\{\Phi(3\mathbf{v}_x*\mathbf{v}_y^2)\}_{1\leq x< y\leq \delta}$, $\{\Phi(3\mathbf{v}_x*\mathbf{v}_y*\mathbf{v}_z)\}_{1\leq x< y<z\leq \delta}$
are all zero vectors. Considering $\gamma+3\delta$ vectors $\{\Phi(\mathbf{u}_i)\}_{i=1}^\gamma$, $\{\Phi(\mathbf{v}_j)\}_{j=1}^\delta$, $\{\Phi(3\mathbf{v}_j)\}_{j=1}^\delta$, $\{\Phi(3\mathbf{v}_j^2)\}_{j=1}^\delta$, we claim that there are $\gamma+2\delta+\overline{r}$ vectors from them are linear independent. Namely, $rank(C)=\gamma+2\delta+\overline{r}$. In fact, $\{\Phi(\mathbf{u}_i)\}_{i=1}^\gamma$, $\{\Phi(\mathbf{v}_j)\}_{j=1}^\delta$, $\{\Phi(3\mathbf{v}_j)\}_{j=1}^\delta$ are linear independent
by the form of $\mathcal{G}$,
and there are $\overline{r}$ vectors in $\{\Phi(3\mathbf{v}_j^2)\}_{j=1}^\delta$ that are linear independent to all these $\gamma+2\delta$ vectors. Without loss of generality, assume $S'=\mathbf{0}$ in $\mathcal{G}$.
Note that $\{\Phi(3\mathbf{v}_j^2)\}_{j=1}^\delta$ are linear independent. In all these $\delta$ vectors $\mathbf{v}_j$,
there are exactly $\overline{r}$ vectors with the form
$$\mathbf{v}_{j'}=(0,\ldots,0,2,0,\ldots,0,1,0,\ldots,0),$$ while other $\delta-\overline{r}$ vectors are with the form $$\mathbf{v}_{j''}=(0,\ldots,0,0,0,\ldots,0,1,0,\ldots,0).$$
As to these $\overline{r}$ vectors, we have $$3\mathbf{v}_{j'}=(0,\ldots,0,6,0,\ldots,0,3,0,\ldots,0), 3\mathbf{v}_{j'}^2=(0,\ldots,0,3,0,\ldots,0,3,0,\ldots,0),$$
thus $\mathbf{v}_j,3\mathbf{v}_j$ and $3\mathbf{v}_j^2$ are linear independent under the Gray map $\Phi$.
However, other $\delta-\overline{r}$ vectors are linear independent.

$(2)$ If $\overline{r}\in \bigg\{\delta+1,\ldots,\delta+\dbinom{\delta}{2}\bigg\}$,
then choosing all $\delta$ columns of $\bar{A}$ and $\overline{r}-\delta$ columns from $\bar{B}$
as $\overline{r}$ columns of $S_r$, and the remaining columns of $S_r$ are all-zero vectors. Namely,
$$S_r=(\overline{A}|\bar{\bar{B}}|\mathbf{0}),$$ where $\bar{\bar{B}}$ is a matrix consisting of $\overline{r}-\delta$ vectors.
In this case, $\{\Phi(3\mathbf{v}_x*\mathbf{v}_y*\mathbf{v}_z)\}_{1\leq x< y<z\leq \delta}$ are all zero vectors, and
$\{\Phi(3\mathbf{v}_x^2*\mathbf{v}_y)\}_{1\leq x< y\leq \delta}$, $\{\Phi(3\mathbf{v}_x*\mathbf{v}_y^2)\}_{1\leq x< y\leq \delta}$
are the same as $\{\Phi(3\mathbf{v}_k*\mathbf{v}_l)\}_{1\leq l< k\leq \delta}$ because of the form of $\bar{B}$.
From Case (1), $\{\Phi(\mathbf{u}_i)\}_{i=1}^\gamma$, $\{\Phi(\mathbf{v}_j)\}_{j=1}^\delta$, $\{\Phi(3\mathbf{v}_j)\}_{j=1}^\delta$, $\{\Phi(3\mathbf{v}_j^2)\}_{j=1}^\delta$
are linear independent.

Next, it also need to prove that there are $\overline{r}-\delta$ linear independent vectors in $\{\Phi(3\mathbf{v}_k*\mathbf{v}_l)\}_{1\leq l< k\leq \delta}$ such that
they are also linear independent to $\{\Phi(\mathbf{u}_i)\}_{i=1}^\gamma$, $\{\Phi(\mathbf{v}_j)\}_{j=1}^\delta$, $\{\Phi(3\mathbf{v}_j)\}_{j=1}^\delta$, $\{\Phi(3\mathbf{v}_j^2)\}_{j=1}^\delta$. By the definition of $\bar{B}$, for fixed $1\leq k< l\leq \delta$, if $\mathbf{h}_k+\mathbf{h}_l$ is in $S_r$,
we denote $\mathbf{h}_k+\mathbf{h}_l$ in the $r'$-th column,
then the $r'$-th coordinate of $3(\mathbf{v}_k*\mathbf{v}_l)$ is $3$ and other coordinates are all zeros.
Therefore, each $\Phi(3(\mathbf{v}_k*\mathbf{v}_l))$ is independent to vectors
$\{\Phi(\mathbf{u}_i)\}_{i=1}^\gamma$, $\{\Phi(\mathbf{v}_j)\}_{j=1}^\delta$, $\{\Phi(3\mathbf{v}_j)\}_{j=1}^\delta$, $\{\Phi(3\mathbf{v}_j^2)\}_{j=1}^\delta$
and $\{\Phi(3(\mathbf{v}_s*\mathbf{v}_t))\},\{s,t\}\neq\{k,l\}$.

$(3)$ If $\overline{r}\in \bigg\{\delta+\dbinom{\delta}{2}+1,\ldots,\delta+2\dbinom{\delta}{2}\bigg\}$,
then choosing all $\delta+\dbinom{\delta}{2}$ columns from the matrices $\bar{A},\bar{B}$, and choosing $\overline{r}-\delta-\dbinom{\delta}{2}$ columns from the matrix $\bar{C}$
as $\overline{r}$ columns of $S_r$, and the remaining columns of $S_r$ are all-zero vectors. Namely,
$$S_r=(\overline{A}|\overline{B}|\bar{\bar{C}}|\mathbf{0}),$$ where $\bar{\bar{C}}$ is a matrix consisting of $\overline{r}-\delta-\dbinom{\delta}{2}$ vectors.

In this case, $\{\Phi(3\mathbf{v}_x*\mathbf{v}_y*\mathbf{v}_z)\}_{1\leq x< y<z\leq \delta}$ are all zero vectors, and
$\{\Phi(3\mathbf{v}_x^2*\mathbf{v}_y)\}_{1\leq x< y\leq \delta}$ are the same as $\{\Phi(3\mathbf{v}_x*\mathbf{v}_y^2)\}_{1\leq x< y\leq \delta}$
by the form of the matrices $\bar{B},\bar{C}$.
Similar to Case $(2)$, by the definition of $\bar{C}$ and the form of $S_r$, all vectors
$\{\Phi(\mathbf{u}_i)\}_{i=1}^\gamma$, $\{\Phi(\mathbf{v}_j)\}_{j=1}^\delta$, $\{\Phi(3\mathbf{v}_j)\}_{j=1}^\delta$, $\{\Phi(3\mathbf{v}_j^2)\}_{j=1}^\delta$, $\{\Phi(3\mathbf{v}_k*\mathbf{v}_l)\}_{1\leq l< k\leq \delta}$
are linear independent, and there are $\overline{r}-\delta-\dbinom{\delta}{2}$ linear independent vectors in $\{\Phi(3\mathbf{v}_x^2*\mathbf{v}_y)\}_{1\leq x< y\leq \delta}$
that are linear independent to all these vectors.

$(4)$ If $\overline{r}\in \bigg\{\delta+2\dbinom{\delta}{2}+1,\ldots,\delta+3\dbinom{\delta}{2}\bigg\}$,
then choosing all $\delta+2\dbinom{\delta}{2}$ columns of matrices $\bar{A},\bar{B},\bar{C}$, and $\overline{r}-\delta-2\dbinom{\delta}{2}$ column vectors from $\bar{D}$
as $\overline{r}$ columns of $S_r$, and the remaining columns of $S_r$ are all-zero vectors.
$$S_r=(\overline{A}|\overline{B}|\overline{C}|\bar{\bar{D}}|\mathbf{0}),$$ where $\bar{\bar{D}}$ is a matrix consisted of $\overline{r}-\delta-2\dbinom{\delta}{2}$ vectors.

In this case, $\{\Phi(3\mathbf{v}_x*\mathbf{v}_y*\mathbf{v}_z)\}_{1\leq x< y<z\leq \delta}$ are still all zero vectors,
but $\{\Phi(3\mathbf{v}_x^2*\mathbf{v}_y)\}_{1\leq x< y\leq \delta}$ are not the same as $\{\Phi(3\mathbf{v}_x*\mathbf{v}_y^2)\}_{1\leq x< y\leq \delta}$
by the form of $\bar{D}$.
Similar to cases $(2)$ and $(3)$, all vectors
$\{\Phi(\mathbf{u}_i)\}_{i=1}^\gamma$, $\{\Phi(\mathbf{v}_j)\}_{j=1}^\delta$, $\{\Phi(3\mathbf{v}_j)\}_{j=1}^\delta$, $\{\Phi(3\mathbf{v}_j^2)\}_{j=1}^\delta$, $\{\Phi(3\mathbf{v}_k*\mathbf{v}_l)\}_{1\leq l< k\leq \delta}$, $\{\Phi(3\mathbf{v}_x^2*\mathbf{v}_y)\}_{1\leq x< y\leq \delta}$
are linear independent,
and there are $\overline{r}-\delta-2\dbinom{\delta}{2}$ linear independent vectors in $\{\Phi(3\mathbf{v}_x*\mathbf{v}_y^2)\}_{1\leq x< y\leq \delta}$
that are linear independent to all these vectors.

$(5)$ If $\overline{r}\in \bigg\{\delta+3\dbinom{\delta}{2}+1,\ldots,\delta+3\dbinom{\delta}{2}+\dbinom{\delta}{3}\bigg\}$,
then choosing all $\delta+3\dbinom{\delta}{2}$ columns of the matrices $\bar{A},\bar{B},\bar{C},\bar{D}$, and $\overline{r}-\delta-3\dbinom{\delta}{2}$ columns from the matrix $\bar{E}$
as $\overline{r}$ columns of $S_r$, and the remaining columns of $S_r$ are all-zero vectors. Namely,
$$S_r=(\overline{A}|\overline{B}|\overline{C}|\overline{D}|\bar{\bar{E}}|\mathbf{0}),$$ where $\bar{\bar{E}}$ is a matrix consisting of $\overline{r}-\delta-3\dbinom{\delta}{2}$ vectors.

In this case, $\{\Phi(3\mathbf{v}_x*\mathbf{v}_y*\mathbf{v}_z)\}_{1\leq x< y<z\leq \delta}$ are not all zero vectors,
by the definition of the matrix $E$ and similar to the discussion of Case (2), all vectors
$\{\Phi(\mathbf{u}_i)\}_{i=1}^\gamma$, $\{\Phi(\mathbf{v}_j)\}_{j=1}^\delta$, $\{\Phi(3\mathbf{v}_j)\}_{j=1}^\delta$, $\{\Phi(3\mathbf{v}_j^2)\}_{j=1}^\delta$, $\{\Phi(3\mathbf{v}_k*\mathbf{v}_l)\}_{1\leq l< k\leq \delta}$, $\{\Phi(3\mathbf{v}_x^2*\mathbf{v}_y)\}_{1\leq x< y\leq \delta}$,
$\{\Phi(3\mathbf{v}_x*\mathbf{v}_y^2)\}_{1\leq x< y\leq \delta}$
are linear independent,
and there are $\overline{r}-\delta-3\dbinom{\delta}{2}$ linear independent vectors in $\{\Phi(3\mathbf{v}_x*\mathbf{v}_y*\mathbf{v}_z)\}_{1\leq x< y<z\leq \delta}$
that are linear independent to all these vectors.

Hence, using this construction, we can get a $\Z_3\Z_{9}$-linear code for each value of $r$.
\end{proof}

Now we give an example to illustrate the main result of Theorem \ref{th:7}. It has been checked by the MAGMA program.
\begin{exmpl}
Consider a $\mathbb{Z}_3 \mathbb{Z}_{9}$-linear code $C$ of type $(\alpha,\beta;2,4;1)$, the generator matrix of $\mathcal{C}=\Phi^{-1}(C)$ is
$$\mathcal{G}=\left(\begin{array}{cc|ccc}1 & T' & \mathbf{0} & 0 & \mathbf{0}\\
0 & \mathbf{0} & 3T_1 & 3 & \mathbf{0}\\
\hline
\mathbf{0} & S' & S_r & \mathbf{0} & I_4
\end{array}\right),$$
where $r\in\{10,11,\ldots,\min(\beta+5,36)\}$ and
$\overline{r}\in\{0,1,\ldots,\min(\beta-5,26)\}$.

According to Theorem \ref{th:7}, let

\begin{center}
$\bar{A}=\begin{pmatrix} 2 & 0 & 0 & 0 \\ 0 & 2 & 0 & 0\\ 0 & 0 & 2 & 0\\ 0 & 0 & 0 & 2 \end{pmatrix}$,
$\bar{B}=\begin{pmatrix} 1 & 1 & 1 & 0 & 0 & 0 \\ 1 & 0 & 0 & 1 & 1 & 0 \\ 0 & 1 & 0 & 0 & 1 & 1 \\ 0 & 0 & 1 & 1 & 0 & 1 \end{pmatrix}$,
$\bar{C}=\begin{pmatrix} 2 & 2 & 2 & 0 & 0 & 0 \\ 2 & 0 & 0 & 2 & 2 & 0 \\ 0 & 2 & 0 & 0 & 2 & 2 \\ 0 & 0 & 2 & 2 & 0 & 2 \end{pmatrix},$

$\bar{D}=\begin{pmatrix} 1 & 1 & 1 & 0 & 0 & 0 \\ 2 & 0 & 0 & 1 & 1 & 0 \\ 0 & 2 & 0 & 0 & 2 & 1 \\ 0 & 0 & 2 & 2 & 0 & 2 \end{pmatrix},$
$\bar{E}=\begin{pmatrix} 1 & 1 & 1 & 0 \\ 1 & 1 & 0 & 1\\ 1 & 0 & 1 & 1\\ 0 & 1 & 1 & 1 \end{pmatrix}$,
\end{center}
and $\mathbf{a}_{j_1},\mathbf{b}_{j_2},\mathbf{c}_{j_3},\mathbf{d}_{j_4},\mathbf{e}_{j_5}$ be the column vectors in $\bar{A},\bar{B},\bar{C},\bar{D},\bar{E}$, respectively,
where $j_1,j_5\in\{1,2,3,4\}$, $j_2,j_3,j_4\in\{1,2,3,4,5,6\}$.
If $\beta=10$, then $r\in\{10,11,12,13,14,15\}$.
Let $S_{10}=\mathbf{0}$, then we have
$$S_{11}=(a_1|\mathbf{0}),
S_{12}=(a_1,a_2|\mathbf{0}),
S_{13}=(a_1,a_2,a_3|\mathbf{0}),$$
$$S_{14}=(a_1,a_2,a_3,a_4|\mathbf{0}),
S_{15}=(\bar{A}|b_1|\mathbf{0}).$$
Then we can get $6$ $\mathbb{Z}_3 \mathbb{Z}_{9}$-linear codes.
That is, for any fixed $\beta$, we can construct $S_r$ by choosing some columns in order from matrices
$\bar{A},\bar{B},\bar{C},\bar{D},\bar{E}$ and other columns are zero vectors.
\end{exmpl}

\begin{remark}
The construction of $S_r$ in the proof of Theorem \ref{th:7} is not unique.
In fact, $S_r$ can be construceted by choosing any $\overline{r}$ different columns from the matrix $M=(\bar{A}|\bar{B}|\bar{C}|\bar{D}|\bar{E})$ with size $\delta\times \bigg(\delta+3\dbinom{\delta}{2}+\dbinom{\delta}{3}\bigg)$.
\end{remark}
\section{Kernel dimension of $\mathbb{Z}_p \mathbb{Z}_{p^2}$-additive codes}\label{sec:4}

In this section, we will explore the dimension of the kernel of $\mathbb{Z}_p\mathbb{Z}_{p^2}$-linear codes $C=\Phi(\mathcal{C})$.

\begin{lemma}\label{l:9}
Let $\mathcal{C}$ be a $\mathbb{Z}_p \mathbb{Z}_{p^2}$-additive code and $C=\Phi(\mathcal{C})$. Then,
$$K(C)=\{\Phi(\mathbf{u})|\mathbf{u}\in\mathcal{C}\text{ \rm{and} } pP'(\mathbf{u},\mathbf{v})\in \mathcal{C}, \forall \mathbf{v}\in\mathcal{C}\}.$$
\end{lemma}
\begin{proof}
The result follows from Lemma \ref{l:4} and the definition of $K(C)$.
\end{proof}

It's obvious that all codewords of order $p$ in $\C$ under $\Phi$ belong to $K(C)$,
and $\Phi(\mathbf{u})\notin K(C)$ if and only if there is $ \mathbf{v}\in \mathcal{C}$ such that $pP'(\mathbf{u},\mathbf{v})\notin \mathcal{C}$.


\begin{lemma}\label{l:10}
Let $\C$ be a $\Z_p \Z_{p^2}$-additive code and $C=\Phi(\C)$.
If $\Phi(\mathbf{x}_1)+\Phi(\mathbf{x}_2)+\ldots+\Phi(\mathbf{x}_n)\in K(C)$ for $\mathbf{x}_1,\mathbf{x}_2,\ldots,\mathbf{x}_n\in \C$,
then $\Phi(\mathbf{x}_1+\mathbf{x}_2+\ldots+\mathbf{x}_n)\in K(C)$.
\end{lemma}
\begin{proof}
It is enough to prove that $n=2$. The general $n$ can be obtained by induction.
By Lemma \ref{l:2}, we have
\begin{eqnarray*}
  \Phi(\x_1+\x_2+pP'(\x_1,\x_2)) &=& \Phi(\x_1+\x_2)+\Phi(pP'(\x_1+\x_2)) \\
   &=& \Phi(\x_1)+\Phi(\x_2)+p\Phi(pP'(\x_1,\x_2))=\Phi(\x_1)+\Phi(\x_2).
\end{eqnarray*}
By Lemma \ref{l:9}, if $\Phi(\x_1+\x_2+pP'(\x_1,\x_2))\in K(C)$, then for all $\vv\in \C$, we have $pP'((\x_1+\x_2+pP'(\x_1,\x_2)),\vv)=pP'((\x_1+\x_2),\vv)\in \C$, that is,
$\Phi(\x_1+\x_2)\in K(C)$.
\end{proof}

From this lemma, we can conclude that for $\x\in \C$ and $a\in \Z_{p^2}$, if $a\Phi(\x)\in K(C)$, then $\Phi(a\x)\in K(C)$,
but the inverse does not have to be true unless the order of $\mathbf{x}$ is $p$.

\begin{proposition}\label{prop:11}
Let $\C$ be a $\Z_p\Z_{p^2}$-additive code of type $(\alpha,\beta;\gamma,\delta;\kappa)$, with generator matrix $\G$, and $C=\Phi(\C)$ with $ker(C)=\gamma+2\delta-\overline{k}$, where $\overline{k}\in \{1,2,\ldots,\delta\}$.
Then, there exists a set $\{\mathbf{v}_1,\mathbf{v}_2,\ldots,\mathbf{v}_{\overline{k}}\}$ of row vectors of order $p^2$ in $\G$ and $\Phi(\vv_i)\notin K(C)$, such that for all $a_i \in \Z_p$,
$$C=\bigcup_{a_i\in \Z_p}\bigg(K(C)+\Phi\bigg(\sum_{i=1}^{\overline{k}}a_i\mathbf{v}_i\bigg)\bigg).$$
\end{proposition}

\begin{proof}
Denote $\tilde{\mathbf{v}}=\sum_{i=1}^{\overline{k}}a_i\mathbf{v}_i, a_i\in \Z_p$.
It is obvious that $C$ can be written as the union of some cosets of $K(C)$.
Since $|C|=p^{\gamma+2\delta}$ and $ker(C)=\gamma+2\delta-\overline{k}$, we know that there are $p^{\overline{k}}$ such cosets.
Firstly, if $\tilde{\mathbf{v}}\neq 0$, then $\Phi(\tilde{\mathbf{v}})\notin K(C)$.
Let $\mathbf{u}_1,\ldots,\mathbf{u}_\gamma,\mathbf{v}_1,\ldots,\mathbf{v}_\delta$ be $\gamma$ and $\delta$ rows in $\G$ of order $p$ and $p^2$, respectively.
By Lemma \ref{l:9}, all codewords of order $p$ in $\C$ under $\Phi$ belong to $K(C)$.
So, there are $p^{\gamma+\delta}$ codewords of order $p$ in $K(C)$ generated by $\gamma+\delta$ codewords,
and there are $\delta-\overline{k}$ codewords $\omega_i,i\in\{1,2,\ldots,\delta-\overline{k}\}$ of order $p^2$ such that $\Phi(\omega_i)\in K(C)$.
It is obvious that $\Phi(\mathbf{u}_1),\ldots,\Phi(\mathbf{u}_\gamma),\Phi(p\mathbf{v}_1),\ldots,\Phi(p\mathbf{v}_\delta)$
,$\Phi(\omega_1),\ldots,\Phi(\omega_{\delta-\overline{k}})$ are linear independent,
and the code $\C$ can also be generated by $\mathbf{u}_1,\ldots,\mathbf{u}_\gamma,\omega_1,\ldots,\omega_{\delta-\overline{k}},\mathbf{v}_{i_1},\ldots,\mathbf{v}_{i_{\overline{k}}}$,
where $\{i_1,i_2,\ldots,i_{\overline{k}}\}\subseteq \{1,2,\ldots,\delta\}$.
Without loss of generality, assume that $\mathbf{v}_{i_1},\ldots,\mathbf{v}_{i_{\overline{k}}}$ are the $\overline{k}$ row vectors $\mathbf{v}_{1},\ldots,\mathbf{v}_{\overline{k}}$ in $\G$.
If $\Phi(\tilde{\mathbf{v}})\in K(C)$, then $\Phi(\mathbf{u}_1),\ldots,\Phi(\mathbf{u}_\gamma),\Phi(p\mathbf{v}_1),\ldots,\Phi(p\mathbf{v}_\delta)$
,$\Phi(\omega_1),\ldots,\Phi(\omega_{\delta-\overline{k}})$ and $\Phi(\tilde{\mathbf{v}})$ are linearly independent, which contradicts with the size of $K(C)$.
That is, $\Phi(\tilde{\mathbf{v}})\in K(C)$ implies that $\tilde{\mathbf{v}}=0.$
Then, it's sufficient to prove that $p^{\overline{k}}-1$ nonzero vectors $\Phi(\sum_{i=1}^{\overline{k}}a_i\mathbf{v}_i)$ are from different cosets, where $a_i\in \Z_p$.
If $\tilde{\mathbf{v}}_1\neq \tilde{\mathbf{v}}_2$, then $\Phi(\tilde{\mathbf{v}}_1)\neq \Phi(\tilde{\mathbf{v}}_2)$ because $\Phi$ is injective.
Suppose that $\Phi(\tilde{\mathbf{v}}_1)$ and $\Phi(\tilde{\mathbf{v}}_2)$ are in the same coset, that is, $\Phi(\tilde{\mathbf{v}}_1)=\Phi(\tilde{\mathbf{v}}_2)+K(C)$.
Hence, $\Phi(\tilde{\mathbf{v}}_1)-\Phi(\tilde{\mathbf{v}}_2)\in K(C)$, and then $\Phi(\tilde{\mathbf{v}}_1)+(p^2-1)\Phi(\tilde{\mathbf{v}}_2)\in K(C)$.
Thus $\Phi(\tilde{\mathbf{v}}_1+(p^2-1)\tilde{\mathbf{v}}_2)\in K(C)$ by Lemma \ref{l:10}.
Thus, $\tilde{\mathbf{v}}_1+(p^2-1)\tilde{\mathbf{v}}_2=\tilde{\mathbf{v}}_1-\tilde{\mathbf{v}}_2=0$, and then $\tilde{\mathbf{v}}_1=\tilde{\mathbf{v}}_2.$
\end{proof}

Note that if $C$ is a $\Z_p\Z_{p^2}$-linear code, then $K(C)$ is a $\Z_p\Z_{p^2}$-linear subcode of $C$ by Lemma \ref{l:10}.
We define the \emph{kernel} of a $\Z_p\Z_{p^2}$-additive code $\C$ of type $(\alpha,\beta;\gamma,\delta;\kappa)$ as $\K(\C)=\Phi^{-1}(K(C))$. According to Lemma \ref{l:9}, $\K(\C)=\{\mathbf{u}\in \C|pP'(\mathbf{u},\vv)\in \C \text{ for all }\mathbf{v}\in \C\}$.
It's easy to check that $\K(\C)$ is a $\Z_p\Z_{p^2}$-additive subcode of $\C$ of type $(\alpha,\beta;\gamma+\overline{k},\delta-\overline{k};\kappa)$.

By Proposition \ref{prop:11}, if $\C$ is a $\Z_p\Z_{p^2}$-additive code with generator matrix $\G$,
then there exists a set $\{\mathbf{v}_1,\mathbf{v}_2,\ldots,\mathbf{v}_{\overline{k}}\}$ of row vectors of order $p^2$ in $\G$, such that for all $a_i \in \Z_p$,
$$\C=\bigcup_{a_i\in \Z_p}\bigg(\K(\C)+\sum_{i=1}^{\overline{k}}a_i\mathbf{v}_i\bigg).$$

\begin{lemma}\label{l:12}
Let $C$ be a $\Z_p\Z_{p^2}$-linear code of type $(\alpha,\beta;\gamma,\delta;\kappa)$ and length $n=\alpha+p\beta$.
Then, $\ker(C)\in\{\gamma+\delta,\gamma+\delta+1,\ldots,\gamma+2\delta\}$.
\end{lemma}
\begin{proof}
Because all codewords of $\C$ with order $p$ under $\Phi$ belong to $K(C)$, thus $ker(C)\geq\gamma+\delta$.
If $C$ is linear, then $ker(C)=\gamma+2\delta$. This completes the proof.
\end{proof}

\begin{theorem}\label{Th:13}
Let $\alpha,\beta,\gamma,\delta,\kappa$ be integers satisfying (\ref{eq:1}), then we have the following two assertions:
\begin{enumerate}
  \item [\rm (i)] Let $C$ be a $\mathbb{Z}_p \mathbb{Z}_{p^2}$-linear code $C$ of type $(\alpha,\beta;\gamma,\delta;\kappa)$,
then $ker(C)=\gamma+2\delta-\overline{k}$, where $\overline{k} \in\{0,1,\ldots,\delta\}$.
  \item [\rm (ii)] If $\overline{k} \in\{0,1,\ldots,\delta\}$,
 then there exists a $\mathbb{Z}_p\mathbb{Z}_{p^2}$-linear code $C$ of type $(\alpha,\beta;\gamma,\delta;\kappa)$ with $ker(C)=\gamma+2\delta-\overline{k}$.
  \end{enumerate}
\end{theorem}
\begin{proof}
Let $\C$ be a $\Z_p\Z_{p^2}$-additive code of type $(\alpha,\beta;\gamma,\delta;\kappa)$ with generator matrix
$$\mathcal{G}=\left(\begin{array}{cc|ccc}I_\kappa & T' & \mathbf{0} & \mathbf{0} & \mathbf{0}\\
\mathbf{0} & \mathbf{0} & \mathbf{0} & pI_{\gamma-\kappa} & \mathbf{0}\\
\hline
\mathbf{0} & S' & S_k & \mathbf{0} & I_\delta
\end{array}\right),$$
where $S_k\in \Z_{p^2}^{\delta\times s}$ for $s=\beta-(\gamma-\kappa)-\delta$, and $\{\mathbf{u}_i\}_{i=1}^\gamma$, $\{\mathbf{v}_j\}_{j=1}^\delta$ are the row vectors of the matrix $\G$ with order $p$ and $p^2$, respectively.
Let $\{\mu_i\}_{i=1}^{\gamma+\delta}$ be the rows of $\G$.
In fact, $\mu_i=\mathbf{u}_i$ for $i\in \{1,\ldots,\gamma\}$ and $\mu_i=\mathbf{v}_j$ for $i\in\{\gamma+1,\ldots,\gamma+\delta\}$ and $j\in \{1,\ldots,\delta\}$.
Let $C=\Phi(\C)$ with $ker(C)=k=\gamma+2\delta-\overline{k}$. Then (i) follows from Lemma \ref{l:12}.
For (ii), it is sufficient to prove that if the conditions in this theorem hold,
then we can construct a $\Z_p\Z_{p^2}$-additive code $\C$ such that $C=\Phi(\C)$ with $ker(C)=\gamma+2\delta-\overline{k}$.
According to the structure of the matrix $\mathcal{G}$, we can construct $\C$ by constructing $S_k$ as follows.

(1) If $\overline{k}=0$, the result is trivial.

(2) If $\overline{k}\geq 1$, setting all elements in the first $\overline{k}$ rows of $S_k$ is $p-1$, and zeros elsewhere.

$$S_k=\begin{array}{@{}c@{}c@{}c@{}c@{}c@{}c@{}}
  \left(\begin{array}{c} p-1 \\ \vdots \\ p-1\\ 0 \\ \vdots \\ 0 \end{array}\right.
& \begin{array}{c} p-1 \\ \vdots \\ p-1\\ 0 \\ \vdots \\ 0 \end{array}
& \begin{array}{c} \cdots \\ \vdots \\ \cdots\\ \cdots \\ \vdots \\ \cdots\end{array}
& \begin{array}{c} p-1 \\ \vdots \\ p-1\\ 0 \\ \vdots \\ 0 \end{array}
& \left.\begin{array}{c} p-1 \\ \vdots \\ p-1\\ 0 \\ \vdots \\ 0 \end{array}\right)
& \begin{array}{c} 1 \\ \vdots \\ \overline{k}\\ \overline{k}+1 \\ \vdots \\ \delta \end{array}
\end{array}$$
Then we will show $ker(C)=\gamma+2\delta-\overline{k}$.
For $\mathbf{c}\in\C$, let $ord(\mathbf{c})$ mean the order of the codeword $\mathbf{c}$, now consider the following two sets $A$ and $B$,
$$A=\{\mathbf{c}|ord(\mathbf{c})=p\}$$ and

$$\begin{aligned}
B=\big\{&\mathbf{c}=\sum_{i=\gamma+\overline{k}+1}^{\gamma+\delta} a_i\mu_i+\sum_{j=1}^{\gamma+\overline{k}} b_j\mu_j\mid a_i\in \Z_{p^2} \text{ and }\exists a_i \in\Z_{p^2}\backslash p\Z_{p^2};\\
&b_j\in \Z_{p} \text{ for } j\in \{1,\ldots,\gamma\} \text{ and } b_j\in p\Z_{p^2} \text{ for } j\in \{\gamma+1,\ldots,\gamma+\overline{k}\}\big\}.
\end{aligned}$$

It is easy to check that the order of $\mathbf{c}\in B$ is $p^2$.
The size of $A$ is $p^{\gamma+\delta}$ and $A\subseteq K(C)$ by Lemma \ref{l:9}.
The size of $B$ is $p^{\gamma+\overline{k}}({p^{2\delta-2\overline{k}}-p^{\delta-\overline{k}}})=p^{\gamma+2\delta-\overline{k}}-p^{\gamma+\delta}$.
Let $T=A\cup B$, therefore, the size of $T$ is $p^{\gamma+2\delta-\overline{k}}$.
By the definition of $\C$, it can be rewritten as
$$\begin{aligned}
\C=\big\{\mathbf{c}^\prime=\sum_{i=\gamma+\overline{k}+1}^{\gamma+\delta} a^\prime_i \mu_i+\sum_{j=1}^{\gamma+\overline{k}} b^\prime_j\mu_j|& a_i^\prime\in \Z_{p^2}; b_j^\prime\in \Z_{p} \text{ for } j\in \{1,\ldots,\gamma\} \text{ and }\\
&b_j^\prime\in \Z_{p^2} \text{ for } j\in \{\gamma+1,\ldots,\gamma+\overline{k}\}\big\}.
\end{aligned}$$
For simplicity, let $\sum_{i=\gamma+\overline{k}+1}^{\gamma+\delta} a_i\mu_i,\sum_{j=1}^{\gamma+\overline{k}} b_j\mu_j,\sum_{i=\gamma+\overline{k}+1}^{\gamma+\delta} a^\prime_i \mu_i,\sum_{j=1}^{\gamma+\overline{k}} b^\prime_j\mu_j$ be $\mathbf{c}_1,\mathbf{c}_2,\mathbf{c}_1^\prime,\mathbf{c}_2^\prime$, respectively.
If $\mathbf{c}=\mathbf{c}_1+\mathbf{c}_2\in B$, then for arbitrary $\mathbf{c}^\prime \in \C$,
we have $pP'(\mathbf{c},\mathbf{c}^\prime)=pP'(\mathbf{c}_1, \mathbf{c}_1^\prime)=\mathbf{0}\in \C$ or the linear combination of $p(p-1)^{-1}\mu_i \in \C$,
where $i\in \{\gamma+\overline{k}+1, \ldots, \gamma+\delta\}$.
Therefore, $B\subseteq K(C)$ by Lemma \ref{l:9}. Thus $T\subseteq K(C)$.

Next we prove that $K(C)= T$. Let $D=\C\backslash T$, and $D$ can be written as follows.
$$\begin{aligned}
D=\big\{\mathbf{d}=\sum_{i=\gamma+\overline{k}+1}^{\gamma+\delta} a_i\mu_i+\sum_{j=1}^{\gamma+\overline{k}} b_j\mu_j\mid &a_i \in\Z_{p^2},
b_j\in \Z_{p} \text{ for } j\in \{1,\ldots,\gamma\} \text{ and }\\
& \exists b_j\in \Z_{p^2}\backslash p\Z_{p^2} \text{ for } j\in \{\gamma+1,\ldots,\gamma+\overline{k}\}\big\}.
\end{aligned}$$
We only consider $s=1$ since all columns of $S_k$ are the same.
For any $\mathbf{d}\in D$, if there exists $ b_{j_1}\in \Z_{p^2}\backslash p\Z_{p^2},j_1\in \{\gamma+1,\ldots,\gamma+\overline{k}\}$, then let $\mathbf{c}^\prime=(p-1)\mu_{j_1}$,
we have $pP'(\mathbf{d},\mathbf{c}^\prime$)=$(\mathbf{0}|0,0,\ldots,0,p(p-1)^{-1},0,\ldots,0)$ or $(\mathbf{0}|p(p-1)^{-1},0,\ldots,0,p(p-1)^{-1},0,\ldots,0)$,
that is, the value of the $(\alpha+1)$-th coordinate is $0$ or $p(p-1)^{-1}$ and the value of the $(\alpha+\beta-\delta+j_1-\gamma)$-th coordinate is $p(p-1)^{-1}$.
Both of these two vectors are not in $\C$ because $(\mathbf{0}|p,0,\ldots,0,p(p-1)^{-1},0,\ldots,0)\in \C$.
Thus there exists a $\mathbf{c}^\prime=p(p-1)^{-1}\mu_{j_1}\in\C$ such that $pP'(\mathbf{d},\mathbf{c}^\prime)\notin \C$  for all $\mathbf{d}\in D$.
Therefore, $D \cap K(C)=\emptyset$, that is, $(\C\backslash T) \cap K(C)=\emptyset$.
Then $T=K(C)$ since $T\subseteq K(C)$.
The assertion (ii) follows.
\end{proof}
\begin{exmpl}
Let $p=3$, and $C$ be a $\Z_3\Z_9$-linear code of type $(\alpha,9;2,5;1)$ with $\ker(C)=12-\overline{k}$. Let $\C=\Phi^{-1}(C)$.
By Theorem \ref{Th:13} (i), $\overline{k}\in\{0,1,2,3,4,5\}$ and $\ker(C)\in\{12,11,10,9,8,7\}$.
Taking the generator matrix of $\C$ as follows:
$$\G_S=\left(\begin{array}{cc|ccc}
1 & T' & \mathbf{0} & 0 & \mathbf{0}\\
0 & \mathbf{0} &\mathbf{0} & 3 & \mathbf{0}\\
\hline
\mathbf{0} & S' & S_k & \mathbf{0} & I_5
\end{array}\right).$$
By Theorem \ref{Th:13} (ii), $S_{12}=(\mathbf{0})$, and $S_{11},S_{10},S_{9},S_{8},S_{7}$ are constructed as follows:
$$
S_{11}=\left(\begin{array}{ccc}
2 & 2 & 2\\
0 & 0 & 0\\
0 & 0 & 0\\
0 & 0 & 0\\
0 & 0 & 0
\end{array}
\right),
S_{10}=\left(\begin{array}{ccc}
2 & 2 & 2\\
2 & 2 & 2\\
0 & 0 & 0\\
0 & 0 & 0\\
0 & 0 & 0
\end{array}
\right),
S_{9}=\left(\begin{array}{ccc}
2 & 2 & 2\\
2 & 2 & 2\\
2 & 2 & 2\\
0 & 0 & 0\\
0 & 0 & 0
\end{array}
\right),$$
$$
S_{8}=\left(\begin{array}{ccc}
2 & 2 & 2\\
2 & 2 & 2\\
2 & 2 & 2\\
2 & 2 & 2\\
0 & 0 & 0
\end{array}
\right),
S_{7}=\left(\begin{array}{ccc}
2 & 2 & 2\\
2 & 2 & 2\\
2 & 2 & 2\\
2 & 2 & 2\\
2 & 2 & 2
\end{array}
\right).
$$
\end{exmpl}
\begin{remark}\label{rem:15}
Considering the construction in Theorem \ref{Th:13}, it is not the only one.
For example, let the matrix
$$S_k=\begin{array}{@{}c@{}c@{}c@{}c@{}c@{}c@{}}
  \left(\begin{array}{c} p-1 \\ \vdots \\ p-1\\ 0 \\ \vdots \\ 0 \end{array}\right.
& \begin{array}{c}  0 \\ \vdots \\ 0\\ 0 \\ \vdots \\ 0 \end{array}
& \begin{array}{c} \cdots \\ \vdots \\ \cdots\\ \cdots \\ \vdots \\ \cdots\end{array}
& \begin{array}{c} 0 \\ \vdots \\ 0\\ 0 \\ \vdots \\ 0 \end{array}
& \left.\begin{array}{c} 0 \\ \vdots \\ 0\\ 0 \\ \vdots \\ 0 \end{array}\right)
& \begin{array}{c} 1 \\ \vdots \\ \overline{k}\\ \overline{k}+1 \\ \vdots \\ \delta \end{array}
\end{array}$$
can also hold, that is, it is enough that there is one column vector which the first $\overline{k}$ coordinates are $2$ and the last $\delta-\overline{k}$ are $0$ in $S_k$.
Other columns just hold the last $\delta-\overline{k}$ are $0$, other coordinates could be any elements in $\Z_9$.
\end{remark}



\section{Pairs of rank and kernel dimension of $\Z_3\Z_{9}$-additive codes}\label{sec:5}

We give the range of the rank of the $\Z_3\Z_{9}$-linear codes if its dimension of the kernel is fixed.
Once a possible pair of values $(r,k)$ is given,
the $\Z_3\Z_{9}$-linear code $C$ of type $(\alpha,\beta;\gamma,\delta;\kappa)$ with $r=rank(C)$ and $k=ker(C)$ can also be constructed.
\begin{lemma}\label{l:15}
Let $\C$ be a $\Z_3\Z_{9}$-additive code of type $(\alpha,\beta;\gamma,\delta;\kappa)$,
and $C=\Phi(\C)$ is the corresponding $\Z_3\Z_{9}$-linear code.
If $rank(C)=\gamma+2\delta+\overline{r}$ and $ker(C)=\gamma+2\delta-\overline{k}$, where $\overline{k}\in\{1,\ldots,\delta\}$,
then $$1\leq\overline{r}\leq \dbinom{\overline{k}}{2}+\dbinom{\overline{k}+2}{3}.$$
Moreover, $\overline{r}=0$ if $\overline{k}=0$.
\end{lemma}

\begin{proof}If $\overline{k}=0$, then $C$ is linear and $\overline{r}=0$.
It is clear that $\overline{r}\geq 1$ if $\overline{k}\geq 1$.
Let $\{\mathbf{u}_i\}_{i=1}^\gamma$ and $\{\mathbf{v}_j\}_{j=1}^\delta$ be the row vectors of the matrix $\G$ with order $p$ and $p^2$, respectively.
By Proposition \ref{prop:11}, there exists a set $\{\mathbf{v}_1,\mathbf{v}_2,\ldots,\mathbf{v}_{\overline{k}}\}$ of row vectors of order $p^2$ in the generator matrix $\G$,
such that for all $a_i \in \Z_p$,
$C=\bigcup_{a_i\in \Z_p}\bigg(K(C)+\Phi\bigg(\sum_{i=1}^{\overline{k}}a_i\mathbf{v}_i\bigg)\bigg).$
According to Lemma \ref{l:2},
$\Phi\bigg(\sum_{i=1}^{\overline{k}}a_i\mathbf{v}_i\bigg)=\Sigma_{j=1}^{\overline{k}}\Phi(\mathbf{v}_j)+\Sigma_{1\leq l \leq k\leq \overline{k}}(\tilde{a}\Phi(3(\mathbf{v}_k\ast v_l))+\Sigma_{1\leq x \leq y\leq z\leq \overline{k}}(\tilde{b}\Phi(3(\mathbf{v}_x\ast \mathbf{v}_y\ast \mathbf{v}_z))$, $\tilde{a},\tilde{b}\in \Z_3$.

From the proof of Theorem \ref{Th:13} (ii), $K(C)=T=A\cup B$.
The set $A$ can be generated by $\{\Phi(\mathbf{u}_i)\}_{i=1}^\gamma$ and $\{\Phi(3\mathbf{v}_j)\}_{j=1}^\delta$.
$$\begin{aligned}
B=\big\{&\mathbf{c}=\sum_{i=\gamma+\overline{k}+1}^{\gamma+\delta} a_i\mu_i+\sum_{j=1}^{\gamma+\overline{k}} b_j\mu_j\mid a_i\in \Z_{9} \text{ and }\exists a_i \in\Z_{9}\backslash 3\Z_{9};\\
&b_j\in \Z_{3} \text{ for } j\in \{1,\ldots,\gamma\} \text{ and } b_j\in 3\Z_{9} \text{ for } j\in \{\gamma+1,\ldots,\gamma+\overline{k}\}\big\}.
\end{aligned}$$
For any vector $\mathbf{c} \in B$, $\Phi(\mathbf{c})=\Phi\bigg(\sum_{i=\gamma+\overline{k}+1}^{\gamma+\delta} a_i\mu_i+\sum_{j=1}^{\gamma+\overline{k}}b_j\mu_j\bigg)
=\Phi\bigg(\sum_{i=\gamma+\overline{k}+1}^{\gamma+\delta} a_i\mu_i\bigg)+\Phi\bigg(\sum_{j=1}^{\gamma+\overline{k}}b_j\mu_j\bigg)$ by Lemma \ref{l:2}.
The order of $\sum_{j=1}^{\gamma+\overline{k}}b_j\mu_j$ is $3$, thus $\Phi\bigg(\sum_{j=1}^{\gamma+\overline{k}}b_j\mu_j\bigg)$
can also be generated by $\{\Phi(\mathbf{u}_i)\}_{i=1}^\gamma$ and $\{\Phi(3\mathbf{v}_j)\}_{j=1}^\delta$.
By the definition of $\mu_j$, now we only consider  $\Phi\bigg(\sum_{i=\gamma+\overline{k}+1}^{\gamma+\delta} a_i\mu_i\bigg)=\Phi\bigg(\sum_{j=\overline{k}+1}^{\delta} a_j\mathbf{v}_j\bigg), a_j \in\Z_9$.

At first, we will prove that if $k\in\{\overline{k}+1,\ldots,\delta\}$,
then for all $\mathbf{v}\in\C$,
$\Phi(3(\mathbf{v}_k*\mathbf{v}+\mathbf{v}_k*\mathbf{v}*\mathbf{v}+\mathbf{v}_k*\mathbf{v}_k*\mathbf{v}))$
is a linear combination of $\{\Phi(\mathbf{u}_i)\}_{i=1}^\gamma$ and $\{\Phi(3\mathbf{v}_j)\}_{j=1}^\delta$.
The condition
$k\in\{\overline{k}+1,\ldots,\delta\}$ implies that $\Phi(\mathbf{v}_k)\in K(C)$ from the proof of Proposition \ref{prop:11}, that is, for all $\mathbf{v}\in \C$,
$3(\mathbf{v}_k*\mathbf{v}+\mathbf{v}_k*\mathbf{v}*\mathbf{v}+\mathbf{v}_k*\mathbf{v}_k*\mathbf{v})\in \C$ by Lemma \ref{l:9}.
Note that all vectors $\{\mathbf{u}_i\}_{i=1}^\gamma$, $\{3\mathbf{v}_j\}_{j=1}^\delta$ and
$\{3(\mathbf{v}_k*\mathbf{v}+\mathbf{v}_k*\mathbf{v}*\mathbf{v}+\mathbf{v}_k*\mathbf{v}_k*\mathbf{v})|k\in\{\overline{k}+1,\ldots,\delta\}\}$ are of order $3$,
thus $\Phi(3(\mathbf{v}_k*\mathbf{v}+\mathbf{v}_k*\mathbf{v}*\mathbf{v}+\mathbf{v}_k*\mathbf{v}_k*\mathbf{v}))$ can be represented linearly by $\{\Phi(\mathbf{u}_i)\}_{i=1}^\gamma$ and $\{\Phi(3\mathbf{v}_j)\}_{j=1}^\delta$ for $k\in\{\overline{k}+1,\ldots,\delta\}$.

Without loss of generality, let $\delta-\overline{k}=2$, and $j_1=\overline{k}+1, j_2=\overline{k}+2$.
Then $\Phi(a_{j_1}\mathbf{v}_{j_1}+a_{j_2}\mathbf{v}_{j_2})=\Phi(a_{j_1}\mathbf{v}_{j_1})+\Phi(a_{j_2}\mathbf{v}_{j_2})+2\Phi(3(\mathbf{v}_{j_1}*\mathbf{v}_{j_2}+\mathbf{v}_{j_1}*\mathbf{v}_{j_1}*\mathbf{v}_{j_2}
+\mathbf{v}_{j_1}*\mathbf{v}_{j_2}*\mathbf{v}_{j_2}))$.
We know $\Phi(3(\mathbf{v}_{j_1}*\mathbf{v}_{j_2}+\mathbf{v}_{j_1}*\mathbf{v}_{j_1}*\mathbf{v}_{j_2}
+\mathbf{v}_{j_1}*\mathbf{v}_{j_2}*\mathbf{v}_{j_2}))$ is a linear combination of $\{\Phi(\mathbf{u}_i)\}_{i=1}^\gamma$ and $\{\Phi(3\mathbf{v}_j)\}_{j=1}^\delta$.
Since $a_{j_1},a_{j_2}\in \Z_9$, it is easy to check that $\Phi(a_{j_1}\mathbf{v}_{j_1}),\Phi(a_{j_2}\mathbf{v}_{j_2})$ can also be generated by
$\{\Phi(\mathbf{u}_i)\}_{i=1}^\gamma$ and $\{\Phi(3\mathbf{v}_j)\}_{j=1}^\delta$. Consequently, $\langle C\rangle$ can be generated by
$\{\Phi(\mathbf{u}_i)\}_{i=1}^\gamma$, $\{\Phi(\mathbf{v}_j)\}_{j=1}^\delta$, $\Sigma_{1\leq l \leq k\leq \overline{k}}\Phi(3(\mathbf{v}_k\ast v_l))$ and
$\Sigma_{1\leq x \leq y\leq z\leq \overline{k}}\Phi(3(\mathbf{v}_x\ast \mathbf{v}_y\ast \mathbf{v}_z))$.
Hence $\overline{r}\leq \dbinom{\overline{k}}{2}+\dbinom{\overline{k}+2}{3}.$
\end{proof}

According to Lemma \ref{l:15}, $rank(C)\in\bigg\{\gamma+2\delta,\ldots,\gamma+2\delta+\dbinom{\overline{k}}{2}+\dbinom{\overline{k}+2}{3}\bigg\}$,
and we have
$$rank(C)\in \bigg\{\gamma+2\delta,...,\rm{min}\bigg(\beta+\gamma+\kappa,\gamma+2\delta+\dbinom{\overline{k}}{2}+\dbinom{\overline{k}+2}{3}\bigg)\bigg\}.$$

\begin{theorem}\label{th:16}
Let $\alpha,\beta,\gamma,\delta,\kappa$ be positive integers satisfying (\ref{eq:1}).
Then, there is a $\mathbb{Z}_3 \mathbb{Z}_{9}$-linear code $C$ of type $(\alpha,\beta;\gamma,\delta;\kappa)$ with $ker(C)=\gamma+2\delta-\overline{k}$ and $rank(C)=\gamma+2\delta+\overline{r}$ if and only if
$\overline{k}\in\{1,\ldots,\delta\}$ and
$\overline{r}\in \bigg\{1,\ldots,\rm{min}\bigg(\beta-(\gamma-\kappa)-\delta,\dbinom{\overline{k}}{2}+\dbinom{\overline{k}+2}{3}\bigg)\bigg\}$
or $\overline{k}=\overline{r}=0$.
\end{theorem}
\begin{proof}
Let $\C$ be a $\Z_3\Z_{9}$-additive code of type $(\alpha,\beta;\gamma,\delta;\kappa)$ with generator matrix
$$\mathcal{G}=\left(\begin{array}{cc|ccc}I_\kappa & T' & \mathbf{0} & \mathbf{0} & \mathbf{0}\\
\mathbf{0} & \mathbf{0} & \mathbf{0} & 3I_{\gamma-\kappa} & \mathbf{0}\\
\hline
\mathbf{0} & S' & S_{r,k} & \mathbf{0} & I_\delta
\end{array}\right),$$
where $S_{r,k}$ is a matrix over $\Z_{9}$ of size $\delta\times (\beta-(\gamma-\kappa)-\delta)$,
and $C=\Phi(\C)$ is its corresponding $\Z_3\Z_{9}$-linear code.
The necessary condition is clear by Lemma \ref{l:15}.
Then we will show the sufficient conditions.

If $\overline{k}=0$, then $C$ is a linear code and $\overline{r}=0$.
If $\overline{k}\geq 1$ and $\overline{r}\geq 1$,
let $\mathbf{h'}_k,1\leq k\leq \delta$ be the column vector of length $\delta$, with a $1$ in the $k$-th coordinate and zeros elsewhere,
similar to those matrices
$\bar{A},\bar{B},\bar{C},\bar{D},\bar{E}$ in the proof of Theorem \ref{th:7},
define matrices $\bar{A'}_{\overline{k} \times \overline{k}}$, $\bar{B'},\bar{C'},\bar{D'}$ of size $\overline{k} \times \dbinom{\overline{k}}{2}$
and $\bar{E'}$ of size $\overline{k} \times \dbinom{\overline{k}}{3}$ over $\Z_9$ as follows:
\begin{itemize}
\item $\bar{A'}_{\overline{k}\times \overline{k}}=2I_{\overline{k}}$;
\item $\bar{B'}=(\mathbf{b'}_j)$, where $\mathbf{b'}_j=\mathbf{h'}_k+\mathbf{h'}_l$ for $1\leq k< l\leq \overline{k} , j\in\bigg\{1,2,\ldots, \dbinom{\overline{k}}{2}\bigg\}$,
and $\mathbf{b'}_{j_1}\neq \mathbf{b'}_{j_2}$ if $j_1 \neq j_2$;
\item $\bar{C'}=(\mathbf{c'}_j)$, where $\mathbf{c'}_1=(2,2,\ldots,2)^T$ and $\mathbf{c'}_j=2\mathbf{b'}_j$ for $j\in\bigg\{2,3,\ldots, \dbinom{\overline{k}}{2}\bigg\}$;
\item $\bar{D'}=(\mathbf{d'}_j)$, where $\mathbf{d'}_j=\mathbf{h'}_k+2\mathbf{h'}_l$ (or equivalently, $\mathbf{d'}_j=2\mathbf{h'}_k+\mathbf{h'}_l$) for $1\leq k< l\leq \overline{k} , j\in\bigg\{1,2,\ldots, \dbinom{\overline{k}}{2}\bigg\}$, and $\mathbf{d'}_{j_1}\neq \mathbf{d'}_{j_2}$ if $j_1 \neq j_2$;
\item $\bar{E'}=(\mathbf{e'}_j)$, where $\mathbf{e'}_j=\mathbf{h'}_x+\mathbf{h'}_y+\mathbf{h'}_z$ for $1\leq x< y<z\leq \overline{k}, j\in\bigg\{1,2,\ldots, \dbinom{\overline{k}}{3}\bigg\}$,
and $\mathbf{e'}_{j_1}\neq \mathbf{e'}_{j_2}$ if $j_1 \neq j_2$.
\end{itemize}

Then $S_{r,k}$ can be constructed in the following way.
Divide $S_{r,k}$ into two matrices $S_{1},S_{2}$, where $S_1$ is of size $\overline{k}\times(\beta-(\gamma-\kappa)-\delta)$,
$S_2$ is of size $(\delta-\overline{k})\times(\beta-(\gamma-\kappa)-\delta)$ such that
$S_{r,k}=\begin{pmatrix}S_1\\S_2\end{pmatrix}$, and let $S_2=(\mathbf{0})$.
Let the first column of $S_1$ be $\mathbf{c'}_1=(2,2,\ldots,2)^T$,
which guarantee $ker(C)=\gamma+2\delta-\overline{k}$ by Theorem \ref{Th:13} and Remark \ref{rem:15}.
As to the remaining columns of $S_1$, similar to the proof of Theorem \ref{th:7},
one can select $\overline{r}$ columns from matrices $\bar{A'},\bar{B'},\bar{C'},\bar{D'},\bar{E'}$ in turn.
Since $\mathbf{c'}_1$ is the first column, one can select the other $\dbinom{\overline{k}}{2}-1$ columns from $\bar{C'}$.
By using the same argument as the proof of Theorem \ref{th:7}, $rank(C)=\gamma+2\delta+\overline{r}$. This completes the proof.
\end{proof}

\begin{exmpl}
Considering a $\mathbb{Z}_3 \mathbb{Z}_9$-linear code $C$ of type $(\alpha,\beta;2,6;1)$.
From Theorem \ref{th:16}, $\overline{k}\in\{0,1,2,3,4,5,6\}$ and $\overline{r}\in\bigg\{0,1,\ldots,\rm{min}\bigg(\beta-7,\dbinom{\overline{k}}{2}+\dbinom{\overline{k}+2}{3}\bigg)\bigg\}$.
Let $\beta=14$, the possible dimensions of the kernel $k\in\{8,9,10,11,12,13,14\}$,
and the possible ranks $r\in\{14,15,16,17,18,19,20,21\}$.
The values of $(r,k)$ are shown in Table \ref{tab:1}.
\begin{table}[!htbp]
\centering
\caption{The values of $(r,k)$}\label{tab:1}
\begin{tabular}{|c|c|c|c|c|c|c|c|c|c|}
\hline
\diagbox{$k$}{$r$}  & 14        & 15        & 16        & 17        & 18        & 19        & 20        & 21          \\ \hline
14                  & $\bullet$ &           &           &           &           &           &           &             \\ \hline
13                  &           & $\bullet$ &           &           &           &           &           &             \\ \hline
12                  &           & $\bullet$ & $\bullet$ & $\bullet$ & $\bullet$ & $\bullet$ &           &             \\ \hline
11                  &           & $\bullet$ & $\bullet$ & $\bullet$ & $\bullet$ & $\bullet$ & $\bullet$ & $\bullet$   \\ \hline
10                  &           & $\bullet$ & $\bullet$ & $\bullet$ & $\bullet$ & $\bullet$ & $\bullet$ & $\bullet$   \\ \hline
9                   &           & $\bullet$ & $\bullet$ & $\bullet$ & $\bullet$ & $\bullet$ & $\bullet$ & $\bullet$   \\ \hline
8                   &           & $\bullet$ & $\bullet$ & $\bullet$ & $\bullet$ & $\bullet$ & $\bullet$ & $\bullet$   \\ \hline
\end{tabular}
\end{table}

The generator matrix of $\mathcal{C}=\Phi^{-1}(C)$ is
$$\mathcal{G}=\left(\begin{array}{cc|ccc}1 & T' & \mathbf{0} & 0 & \mathbf{0}\\
0 & \mathbf{0} & 3T_1 & 3 & \mathbf{0}\\
\hline
\mathbf{0} & S' & S_{r,k} & \mathbf{0} & I_4
\end{array}\right).$$

When $\overline{k}=4$, the related matrices are as follows, others can be obtained similarly and be omitted.
\begin{center}
$\bar{A'}=\begin{pmatrix} 2 & 0 & 0 & 0 \\ 0 & 2 & 0 & 0\\ 0 & 0 & 2 & 0\\ 0 & 0 & 0 & 2 \end{pmatrix}$,
$\bar{B'}=\begin{pmatrix} 1 & 1 & 1 & 0 & 0 & 0 \\ 1 & 0 & 0 & 1 & 1 & 0 \\ 0 & 1 & 0 & 0 & 1 & 1 \\ 0 & 0 & 1 & 1 & 0 & 1 \end{pmatrix}$,
$\bar{C'}=\begin{pmatrix} 2 & 2 & 2 & 0 & 0 & 0 \\ 2 & 0 & 0 & 2 & 2 & 0 \\ 2 & 2 & 0 & 0 & 2 & 2 \\ 2 & 0 & 2 & 2 & 0 & 2 \end{pmatrix},$

$\bar{D'}=\begin{pmatrix} 1 & 1 & 1 & 0 & 0 & 0 \\ 2 & 0 & 0 & 1 & 1 & 0 \\ 0 & 2 & 0 & 0 & 2 & 1 \\ 0 & 0 & 2 & 2 & 0 & 2 \end{pmatrix},$
$\bar{E'}=\begin{pmatrix} 1 & 1 & 1 & 0 \\ 1 & 1 & 0 & 1\\ 1 & 0 & 1 & 1\\ 0 & 1 & 1 & 1 \end{pmatrix}$.

$S_{15,10}=\left(
\begin{array}{@{}c@{}|c}  \begin{array}{c}2\\2\\2\\2\\0\\0\end{array}  & \mathbf{0}\end{array}
\right),$
$S_{16,10}=\left(
\begin{array}{@{}c@{}|c}  \begin{array}{cc}2 & 2\\2& 0\\2& 0\\2& 0\\0& 0\\0& 0\end{array}  & \mathbf{0}\end{array}
\right),$
$S_{17,10}=\left(
\begin{array}{@{}c@{}|c}  \begin{array}{ccc}2 & 2 & 0\\2& 0&2\\2& 0& 0\\2& 0& 0\\0& 0& 0\\0& 0& 0\end{array}  & \mathbf{0}\end{array}
\right),$

$S_{18,10}=\left(
\begin{array}{@{}c@{}|c}  \begin{array}{cccc}2 & 2 & 0&0\\2& 0&2&0\\2& 0& 0&2\\2& 0& 0&0\\0& 0& 0&0\\0& 0& 0&0\end{array}  & \mathbf{0}\end{array}
\right),$
$S_{19,10}=\left(
\begin{array}{@{}c@{}|c}  \begin{array}{ccccc}2 & 2 & 0&0&0\\2& 0&2&0&0\\2& 0& 0&2&0\\2& 0& 0&0&2\\0& 0& 0&0&0\\0& 0& 0&0&0\end{array}  & \mathbf{0}\end{array}
\right),$

$S_{20,10}=\left(
\begin{array}{@{}c@{}|c}  \begin{array}{cccccc}2 & 2 & 0&0&0&1\\2& 0&2&0&0&1\\2& 0& 0&2&0&0\\2& 0& 0&0&2&0\\0& 0& 0&0&0&0\\0& 0& 0&0&0&0\end{array}  & \mathbf{0}\end{array}
\right),$
$S_{21,10}=\left(
\begin{array}{ccccccc}2 & 2 & 0&0&0&1&1\\2& 0&2&0&0&1&0\\2& 0& 0&2&0&0&1\\2& 0& 0&0&2&0&0\\0& 0& 0&0&0&0&0\\0& 0& 0&0&0&0&0\end{array}
\right).$
\end{center}

\end{exmpl}

\section{Conclusion}\label{sec:6}

In this paper, we studied the kernel of $\Z_p\Z_{p^2}$-linear codes and the rank of $\Z_3\Z_{9}$-linear codes.
Using combinatorial enumeration techniques, we gave the lower and upper bounds of the rank of $\Z_3\Z_{9}$-linear codes and dimension of the kernel of $\Z_p\Z_{p^2}$-linear codes, respectively.
Moreover, we have constructed $\Z_p\Z_{p^2}$-linear codes for each value of the dimension of the kernel,
and constructed $\Z_3\Z_{9}$-linear codes for each value of the rank.
Finally, we also constructed $\Z_3\Z_{9}$-linear codes for pairs of values of rank and the dimension of the kernel.

\section{Acknowledgement}

The authors are grateful to Professor Denis Krotov's helpful discussion.

\end{document}